\documentclass[a4paper,USenglish,cleveref]{lipics-v2021}
\nolinenumbers
\usepackage{xspace}
\usepackage{dsfont}
\usepackage{csquotes}
\usepackage{mathtools}
\usepackage{tabto}
\usepackage[ruled,vlined,linesnumbered]{algorithm2e}
\DontPrintSemicolon
\usepackage{contour}
\usepackage{refcount}

\usepackage{tikz}
\usetikzlibrary{backgrounds}
\usetikzlibrary{calc}
\usetikzlibrary{math}
\usetikzlibrary{positioning}
\usetikzlibrary{arrows.meta}
\usetikzlibrary{decorations.pathreplacing}
\usetikzlibrary{decorations.pathmorphing}
\usetikzlibrary{shapes.misc}
\usetikzlibrary{patterns}
\definecolor{superred}{rgb}{0.87,0.1,0.1}
\definecolor{superblue}{RGB}{66,102,170}
\definecolor{indigo}{HTML}{2D2F92}
\definecolor{mygreen}{rgb}{0, 0.75, 0.3}
\newcommand{\yelloworange}{yellow!50!orange}
\newcommand{\lightgray}{gray!15}
\newcommand{\darkgray}{gray!35}

\renewcommand{\cal}[1]{\ensuremath{\mathcal{#1}}\xspace}
\newcommand{\bb}[1]{\ensuremath{\mathbb{#1}}\xspace}
\definecolor{superred}{rgb}{0.87,0.1,0.1}

\newcommand{\1}[1]{\operatorname{\mathds{1}}\!\big( #1 \big)}
\newcommand{\argmax}{\operatorname{argmax}}
\newcommand{\range}[1]{\ensuremath{[#1]}}

\newcommand{\restrict}[1]{|_{#1}}

\newcommand{\apx}{\textsf{\small apx}}
\newcommand{\opt}{\textsf{\small opt}}
\newcommand{\A}{\cal A}
\newcommand{\B}{\cal B}
\newcommand{\C}{\cal C}
\newcommand{\ff}{f^\times}
\newcommand{\G}{\cal G}
\newcommand{\GG}{\bb G}
\newcommand{\M}{\cal M}
\newcommand{\Mt}{\M\restrict t}
\renewcommand{\P}{\cal P}
\newcommand{\qual}{q}
\newcommand{\qualcap}{\qual_\cap}
\newcommand{\Qual}{Q}
\renewcommand{\S}{\cal S}
\renewcommand{\st}{$s$-$t$-}

\newcommand{\siti}{$s_i$-$t_i$-}

\newcommand{\wmod}[1][Y]{w_{\varepsilon,#1}}

\newcommand{\problem}[1]{\textup{\textsf{#1}}\xspace}
\newcommand{\MMinCE}{\problem{MMinCut}}

\newcommand{\MMaxBIS}{\problem{MMaxBIS}}
\newcommand{\MMinBVC}{\problem{MMinBVC}}

\newcommand{\MMaxWBC}{\problem{MWBMaxCut}}
\newcommand{\MMinWPM}{\problem{MMinPM}}
\newcommand{\MSPath}{\problem{MSPath}}

\let\oldchi\chi
\renewcommand{\chi}{\text{$\oldchi$}}
\let\oldsum\sum
\renewcommand{\sum}{{\textstyle\oldsum\nolimits}}
\let\oldbigcup\bigcup
\renewcommand{\bigcup}{{\textstyle\oldbigcup\nolimits}}
\let\oldbigcap\bigcap
\renewcommand{\bigcap}{{\textstyle\oldbigcap\nolimits}}
\renewcommand{\frac}[2]{{\tfrac{#1}{#2}}}
\let\oldexists\exists
\renewcommand{\exists}{\oldexists\,}
\let\oldnexists\nexists
\renewcommand{\nexists}{\oldnexists\,}
\let\oldforall\forall
\renewcommand{\forall}{\oldforall\,}

\title{A General Approximation for Multistage Subgraph Problems}
\titlerunning{A General Approximation for Multistage Subgraph Problems}
\author{Markus Chimani}{Theoretical Computer Science, Osnabrück University, Germany}{markus.chimani@uos.de}{https://orcid.org/0000-0002-4681-5550}{}
\author{Niklas Troost}{Theoretical Computer Science, Osnabrück University, Germany}{niklas.troost@uos.de}{https://orcid.org/0000-0001-7412-2770}{}
\author{Tilo Wiedera}{Theoretical Computer Science, Osnabrück University, Germany}{tilo.wiedera@uos.de}{https://orcid.org/0000-0002-5923-4114}{}
\authorrunning{M. Chimani, N. Troost and T. Wiedera}
\Copyright{M. Chimani, N. Troost and T. Wiedera}
\hideLIPIcs
\begin{document}
	\maketitle

	\begin{abstract}
		\emph{Subgraph Problems} are optimization problems on graphs where a solution is a subgraph that satisfies some property and optimizes some measure.
		Examples include shortest path, minimum cut, maximum matching, or vertex cover.
		In reality, however, one often deals with time-dependent data, i.e., the input graph may change over time and we need to adapt our solution accordingly.
		We are interested in guaranteeing optimal solutions after each graph change while retaining as much of the previous solution as possible.
		Even if the subgraph problem itself is polynomial-time computable, this \emph{multistage} variant turns out to be NP-hard in most cases.

		We present an algorithmic framework that---for any subgraph problem of a certain type---guarantees an optimal solution for each point in time and provides an approximation guarantee for the similarity between subsequent solutions.

		We show that the class of applicable multistage subgraph problems is very rich and that proving membership to this class is mostly straightforward.
		As examples, we explicitly state these proofs and obtain corresponding approximation algorithms for the natural multistage versions of Shortest \st Path, Perfect Matching, Minimum \st Cut---and further classical problems on bipartite or planar graphs, namely Maximum Cut, Vertex Cover, and Independent Set.
		We also report that all these problems are already NP-hard on only two stages.
	\end{abstract}


\section{Introduction \& Related Work}

Subgraph Problems (SPs) are concerned with selecting some feasible set of graph elements (vertices and/or edges) that is optimal w.r.t.\ some measure.
The class of SPs is very rich:
it includes many traditional planning problems like
Shortest \st Path,
Minimum \st Cut,
Minimum Weight Perfect Matching,
Maximum Independent Set,
Minimum Vertex Cover,
Maximum Cut,
Maximum Planar Subgraphs,
Steiner Trees,
etc.

As problem instances may be subject to change over time, it is often required to solve the same SP multiple times. A \emph{multistage graph} is simply a sequence of graphs (the \emph{stages}) and we ask for an individual solution per stage.
A natural concern is to avoid big changes when transitioning from one solution to the next,
since each change might introduce costs for changing the state of the corresponding entities.
Depending on the problem, the transition quality may be measured differently.

Introduced by Gupta et al.~\cite{GTW14} and Eisenstat et al.~\cite{EMS14},
multistage (a.k.a.\ temporal) graph problems have shown to be a rich subject of research.
In many cases, generalizing a polynomial-time solvable problem to a multistage setting renders it NP-hard, e.g., Multistage Shortest \st Path~\cite{FNSZ20} or Multistage Perfect Matching~\cite{GTW14}.
There is some work on identifying parameters that allow for fixed-parameter tractability of NP-hard multistage problems~\cite{FNSZ20,BFK22,FNRZ22,Flu21,HHK+21}.
Another popular approach to tackle such problems are approximation algorithms~\cite{EMS14,BEK21,BELP18,BEST21,BET22}, see~below.

In most of these works,
the objective is to optimize a combined quantity $O + T$ that measures both
the objective value~$O$ of the individual per-stage solutions and
the quality~$T$ of the transitions between subsequent solutions.
For example, consider a multistage variant of the Maximum Matching problem: Given a sequence of~$\tau$ graphs~$(G_i)_{i=1}^\tau$, find a matching~$M_i$ for each~$G_i$ such that~$O+T$ is maximized.
Here, $O\coloneqq \sum_{i=1}^\tau |M_i|$ is the sum of the individual matchings' cardinalities and~$T\coloneqq \sum_{i=1}^{\tau-1} |M_i\cap M_{i+1}|$ is the sum of transition qualities measured as the cardinality of edges that are common between subsequent solutions.

In an approximation setting, this combined objective allows to trade suboptimal transitions for suboptimal solutions in some stages.
This is exploited, e.g., in a $2$-approximation for a multistage Vertex Cover problem~\cite{BEK21}
and a $3$-approximation for a~$3$-stage Minimum Weight Perfect Matching problem on metric graphs~\cite{BELP18}.
In~\cite{BEST21}, several upper and lower bounds for competitive ratios of online algorithms are shown for general Multistage Subset Maximization problems;
their algorithms are not considering running times and depend on polynomial oracles for the underlying single-stage problems.

Contrasting this combined objective, the focus in~\cite{CTW22} is to break up the interdependency between~$O$ and~$T$ for several types of Multistage Perfect Matching problems.
Here, $O$ is required to yield optimal values, i.e., an approximation algorithm must yield optimal solutions for each individual stage.
Thus, the approximation factor truly measures the difficulty in approximating the transition cost.
Regarding exact algorithms, this would only be a special case of the combined objective, where transition costs are scaled appropriately.
However, approximation guarantees are in general \emph{not} transferable to this special case as the approximation may require non-optimal solutions in individual stages, see~\cite{CTW22} for a detailed example.
They also provide several approximation algorithms for multistage perfect matching problems, where the approximation factor is dependent on the maximum size of the intersection between two adjacent stages (we will later define a similar parameter as intertwinement).


\subparagraph{Key contribution.}
In this paper, we provide a framework to obtain approximation algorithms for a wide range of multistage subgraph problems, where, following the concept of~\cite{CTW22}, we guarantee optimal solutions in each stage (cf.\ \cref{sec:framework}).
As a key ingredient we define \emph{preficient} (short for \emph{preference efficient}) problems (Definition~\ref{def:preficient}); they allow to efficiently compute an optimal solution to an individual stage that prefers some given graph elements. As it turns out, many polynomial-time solvable graph problem are in fact trivially preficient. Secondly, we introduce the multistage graph parameter \emph{intertwinement} that provides a measure for the maximum similarity between subsequent stages of the multistage input graph.
Our framework algorithm can be applied to any preficient multistage subgraph problem where we measure the
transition quality as the number of common graph elements between subsequent stages; it yields an approximation ratio only dependent on the input's intertwinement, see Theorem~\ref{thm:2-stage-apx}.

A building block of this algorithm, which itself does not depend on the transition quality measure and may therefore be of independent interest, is Theorem~\ref{thm:tim-to-mim}:
Any $\alpha$-approximation (including $\alpha=1$) for a $t$-stage Subgraph Problem with fixed $t\geq 2$ can be lifted to an approximation for the corresponding unrestricted multistage subgraph problem.

Finally, in \cref{sec:applications}, we demonstrate that the class of applicable multistage problems is very rich: It is typically straightforward to construct a preficiency algorithm from classic algorithms. We can thus deduce several new approximation results simply by applying our preficiency framework approximation, without the need of additional deep proofs.
As examples, we showcase this for multistage variants of Shortest \st Path, Perfect Matching, and Minimum \st Cut. Furthermore, several NP-hard (single-stage) problems become polynomial-time solvable on restricted graph classes (e.g., planar, bipartite, etc.); on these, we can apply our framework as well, as we showcase for Maximum Cut, Vertex Cover, and Independent Set.


\section{Framework}\label{sec:framework}
All known successful applications of our framework are subgraph problems on graphs.
Thus, and for ease of exposition, we will describe our framework solely in this context.
However, we will never use any graph-intrinsic properties other than the fact that it is a system of elements.
It should be understood that we can in fact replace graphs with any other combinatorial structure (e.g., hypergraphs, matroids, fields, etc.) in all definitions and results, as long as a solution is a subset of elements of said structure.

For a graph~$G$, we refer to its vertices and edges collectively as \emph{elements}~$X(G)$.
An \emph{enriched graph} is a graph with additional information (e.g., weights or labels) at its elements.
The following definitions may at first seem overly complicated, but are carefully constructed to be as general as possible, similar to those for general NP optimization problems, e.g.\ in~\cite{KV08, FG06}.

\begin{definition}
	A \emph{Subgraph Problem (SP)} is a combinatorial optimization problem $\P\coloneqq(\GG, f, m, \psi)$, where
	\begin{itemize}
		\item $\GG$ denotes a class of enriched graphs that is the (in general infinite) \emph{set of possible instances};
		\item $f$ is a function such that, for an instance~$G\in\GG$, the set $f(G)\subseteq 2^{X(G)}$ contains the \emph{feasible solutions}; a \emph{feasible solution} $S\in f(G)$ is a subset of~$X(G)$;
		\item $m$ is a function such that, for an instance~$G\in\GG$ and a feasible solution~$S\in f(G)$, the \emph{measure} of~$S$ is given by~$m(G,S)$;
		\item the \emph{goal} $\psi$ is either $\min$ or $\max$.
	\end{itemize}
	Given some instance $G\in \GG$, the objective is to find a feasible solution $S\in f(G)$ that is \emph{optimal} in the sense that $m(G,S) = \psi\{m(G,S') \mid S'\in f(G)\}$.
	The \emph{set of optimal solutions} is denoted by~$f^*(G)$.
	An element $x\in X(G)$ is \emph{allowed w.r.t.~$\P$} if~$x\in X_\P(G)\coloneqq\bigcup_{S\in f^*(G)}S$.
\end{definition}

For $n,n'\in\mathbb N$, let $\range{n,n'}\coloneqq\{n,n+1,...,n'\}$ and $\range{n} \coloneqq\range{1,n}$.
A \emph{multistage graph} (or \emph{$\tau$-stage graph}), is a sequence of graphs $\G=(G_i)_{i\in\range\tau}$ for some $\tau\in\mathbb N_{>0}$.
The graph $G_i \coloneqq (V_i,E_i)$ is the $i$-\emph{th stage of}~$\G$.

\begin{definition}
	A \emph{Multistage Subgraph Problem (MSP)} is a combinatorial optimization problem $\M = (\P, \qual)$, where
	\begin{itemize}
		\item $\P\coloneqq(\GG,f,m,\psi)$ is a Subgraph Problem;
		\item an instance is a multistage graph $\G = (G_i)_{i\in\range\tau}\in\GG^\tau$ for some~$\tau\in\mathbb N_{>0}$; and
		\item $\qual$ is a non-negative function such that, given an instance $\G\in\GG^\tau$ and subsets~$Y_i\subseteq X(G_i)$ and~$Y_{i+1}\subseteq X(G_{i+1})$, $\qual(Y_i,Y_{i+1})$ measures the \emph{transition quality} of these sets for any $i\in\range{\tau-1}$.
	\end{itemize}
	Given some instance $\G\in\GG^\tau$, let $\ff(\G)\coloneqq f^*(G_1) \times ...\times f^*(G_\tau)$ denote the set of \emph{feasible multistage solutions}, containing $\tau$-tuples of optimal solutions for the individual stages.
	The objective is to find a feasible multistage solution $\S\in \ff(\G)$ that is maximum w.r.t.\ $\qual$ in the sense that $\Qual(\S) = \max\{\Qual(\S') \mid \S'\in \ff(\G)\}$
	where $\Qual(\S')\coloneqq \sum_{i\in\range{\tau-1}}\qual(S'_i,S'_{i+1})$ is the \emph{global quality} of~$\S$.
\end{definition}

We stress that a feasible multistage solution must necessarily consist of optimal solutions w.r.t.~$\P$ in each stage.
If there is an upper bound~$t$ on the number of stages $\tau$, an MSP~$\M$ may be denoted by~$\Mt$.
MSPs with some fixed function $\qual$ may be summarily referred to as $\qual$-MSPs.
In our definition, we aim at maximizing the transition quality.
Common choices for transition qualities are the \emph{intersection profit} $\qualcap(S_i,S_{i+1}) \coloneqq |S_i\cap S_{i+1}|$, e.g., in~\cite{FNSZ20,BELP18,BEST21,CTW22},
or measures based on the (symmetric) difference of subsequent stages~\cite{BEK21, BET22, KRZ21}.
Some multistage problems also consider minimizing transition \emph{costs} (see, e.g., \cite{FNSZ20, FNRZ22, Flu21}).


Considering some SP~$\P$ and a multistage graph~$\G=(G_i)_{i\in\range\tau}$,
we can measure the similarity of consecutive stages of~$\G$ w.r.t.~$\P$ by the \emph{intertwinement}
$\text{$\chi$}_\P(\G)\coloneqq \max_{i\in\range{\tau -1}} |X_\P(G_i)\cap X_\P(G_{i+1})|$.
If the context is clear, we may simply use~$\chi\coloneqq\chi_\P(\G)$.

We will establish approximation algorithms whose approximation quality decreases monotonously with increasing~$\chi$.
Consider any MSP~$\M$ with polynomial-time solvable SP~$\P$ and a $2$-stage input graph, where the two stages have nothing in common.
Then, optimizing~$\M$ is typically a simple matter of solving~$\P$ on each stage individually, yielding an exact polynomial-time algorithm.
Observe that the intertwinement captures this as $\chi=0$.
Increasing the commonality between the stages increases both the intertwinement and the multistage problem's complexity, suggesting that intertwinement is a feasible measure.
At some tipping point, once stages become very similar, our use of intertwinement loses its expressiveness and other similarity measures should be preferred.

\section{Algorithmic techniques for approximation}\label{sec:approx}

Given some MSP instance, we denote with $\opt$ its optimal solution value and with $\apx$ the objective value achieved by a given approximation algorithm.
The \emph{approximation ratio} of an approximation algorithm for a maximization problem is the infimum of $\apx / \opt$ over all instances.

\subsection{Reducing the number of stages}\label{subsec:approx-sandwich}

It is natural to expect that an MSP~$\M$ would be easier to solve if the number of stages is bounded by some constant~$t\geq 2$.
We can show that if we have an $\alpha$-approximate (or even an exact) algorithm~$\A$ for~$\Mt$, one can use it to craft a solution for the unbounded problem that is within factor $\alpha\cdot(t-1)/t$ of the optimum.

\newcommand{\Git}[2]{\G\restrict{#2}^{#1}}
\begin{algorithm}[tb]
\caption{
	Approximation of~$\M = (\P, \qual)$ given an $\alpha$-approximation~$\A$ for~$\Mt$\label{algo:tm-to-mim}
}
\KwIn{Enriched multistage graph $\G=(G_1,...,G_\tau)$, $\alpha$-approximation~$\A$ for~$\Mt$}
\KwOut{Multistage solution $\S$}
$\S=(\varnothing,..., \varnothing)$\;
\ForEach{$k\in\range{t}$}{
	$\S_k \gets \A(\Git{k}{1})$\;
	$i\gets k+1$\;
	\While{$i\leq\tau$}{
		$\S_k\gets \S_k\circ \A(\Git{t}{i})$\;
		$i\gets i+t$
	}
	\lIf{$\Qual(\S_k) \geq \Qual(\S)$}{$\S\gets\S_k$}
}
\end{algorithm}

Informally, $\Git{r}{i}$ denotes the subinstance of~$\G$ with~$r$ stages, starting at the $i$-th.
Formally, for~$i\in\range{\tau}$, let $\Git{r}{i}$ consist exactly of the stages with index in the range $\range{i, \min\{i+r-1,\tau\}}$.
Let~$\M$ be an MSP, $t\geq 2$, and assume we have an $\alpha$-approximation algorithm~$\A$ for~$\Mt$.
Algorithm~\ref{algo:tm-to-mim} constructs a set $\{\S_k\mid k\in\range t\}$ of candidate solutions and returns one with maximum global quality.
Each candidate solution~$\S_k$ is built as follows:
Start with a partial solution obtained from calling~$\A$ on the first $k$ stages of~$\G$;
then iteratively consider the subsequent $t$ stages as a subinstance, compute a partial solution using~$\A$ and append it to the existing partial solution (denoted by operator~$\circ$);
repeat this step until eventually stage~$\tau$ has been considered (in general, the final subinstance containing stage~$\tau$ may again consist of less than~$t$ stages).

\begin{theorem}\label{thm:tim-to-mim}
Let $\M = (\P, \qual)$ be an MSP and~$\A$ an $\alpha$-approximation for $\Mt$ for some fixed~$t\geq 2$ (possibly with $\alpha=1$).
Then Algorithm~\ref{algo:tm-to-mim} yields a $\beta$-approximation for~$\M$, where $\beta \coloneqq \alpha (t-1)/t$.
\end{theorem}

\newcommand{\J}{I}
\newcommand{\Skit}[2]{\S_k\restrict{#2}^#1}
\newcommand{\Skoit}[2]{\S^*_k\restrict{#2}^#1}
\begin{proof}
The algorithm's output $\S$ is the candidate with optimal profit and thus has at least average profit over all $t$ candidate solutions:
$\Qual(\S) \geq 1/t\cdot\sum_{k\in\range{t}} \Qual(\S_k)$.
Let $\J_k\coloneqq\{b\leq\tau\mid j\in\mathbb N\colon b=(k+1)+j\cdot t\}$ denote the set of values that $i$ takes in the $k$-th iteration of the foreach loop.
For $r\in\mathbb N_{>0}$,
let $\Qual(\Skit{r}{i}) \coloneqq \sum_{j\in\range{i,i+r-2}} \qual(S_j, S_{j+1})$ denote the quality of~$\S_k\coloneqq(S_j)_{j\in\range\tau}$ restricted to~$\Git{r}{i}$.
As transition qualities are non-negative, we have $\Qual(\S_k) \geq \Qual(\Skit{k}{1}) + \sum_{i\in \J_k}\Qual(\Skit{t}{i})$ and thus $\Qual(\S)\geq 1/t\cdot\sum_{k\in\range{t}}\big(\Qual(\Skit{k}{1}) + \sum_{i\in \J_k}\Qual(\Skit{t}{i})\big)$.

\newcommand{\Qualopt}[2]{\Qual^*\restrict{#2}^{#1}}
Let $\Qualopt{r}{i}$ be the optimal quality achievable for $\Git{r}{i}$.
Since $\A$ is an $\alpha$-approximation for~$\Mt$,
we have for all~$k\in\range t$ that $\Qual(\Skit{k}{1}) \geq \alpha\Qualopt{k}{1}$
and also $\Qual(\Skit{t}{i}) \geq \alpha\Qualopt{t}{i}$ for all $i\in \J_k$.
By construction, $\range{2,\tau}$ is the disjoint union of all $\J_k$ and thus
\[\Qual(\S)
\geq 1/t\cdot\sum_{k\in\range{t}}\big(\alpha\Qualopt{k}{1} + \sum_{i\in \J_k}\alpha\Qualopt{t}{i}\big)
= \alpha/t\cdot\big(\sum_{k\in\range{t}}\Qualopt{k}{1} +  \sum_{i\in\range{2,\tau}}\Qualopt{t}{i}\big).\]

Let $\S^* \coloneqq (S^*_j)_{j\in\range\tau}$ be an optimal multistage solution.
As by definition $\Qualopt{r}{i}\geq \Qual(\Skoit{r}{i})$, we have
$\Qual(\S)
\geq \alpha/t\cdot\big(
\underbrace{\sum_{k\in\range{t}}\Qual(\Skoit{k}{1})}_{(a)}
+ \underbrace{\sum_{i\in\range{2,\tau}}\Qual(\Skoit{t}{i})}_{(b)}
\big)$.

For each fixed~$j\in\range{t-1}$, the term $\qual(S^*_j,S^*_{j+1})$ appears exactly $t-j$ times in $(a)$ (namely once for each $k\in\range{j+1,t}$) and exactly $j-1$ times in $(b)$ (namely once for each $i\in\range{2,j}$).
Considering any larger~$j\in\range{t,\tau-1}$, the term $\qual(S^*_j,S^*_{j+1})$ does not appear in $(a)$ and exactly $t-1$ times in $(b)$ (namely once for each $i\in [j-t+2,j]$).
We thus have
\[\sum_{k\in\range{t}}\Qual(\Skoit{k}{1}) +  \sum_{i\in\range{2,\tau}}\Qual(\Skoit{t}{i})
= (t-1)\cdot\sum_{j\in\range{\tau-1}}\qual(S_j^*, S_{j+1}^*)
= (t-1) \Qual(\S^*)\]
and can conclude
$\Qual(\S) \geq \alpha (t-1)/t\cdot\Qual(\S^*).$
\end{proof}


\subsection{Approximating two stages}\label{subsec:approx-onion}

Generalizing the results and proof techniques of \cite{CTW22}, we present an algorithm that computes an approximate solution for the $2$-stage restriction of any $\qualcap$-MSP where the underlying SP has a certain property.

\begin{definition}\label{def:preficient}
An SP $\P\coloneqq(\mathbb G, f, m, \psi)$ is called \emph{preficient} (short for \emph{preference efficient})
if there is an algorithm~$\B(G,Z)$ that solves the following problem in polynomial time:
Given a graph~$G\in \mathbb G$ and subset $Z\subseteq X(G)$,
compute an optimal solution $S=\argmax_{S'\in f^*(G)} |S'\cap Z|$.
Such an algorithm~$\B$ is called a \emph{preficiency algorithm} for~$\P$.
An MSP is called preficient if its underlying SP is preficient.
\end{definition}

Note that for a preficient SP~$\P$, the set~$X_\P(G)$ is trivially computable in  polynomial time: a~graph element~$x\in X(G)$ is allowed w.r.t.~$\P$ if and only if it is in a solution computed by~$\B(G,\{x\})$.\smallskip

Provided the existence of a preficiency algorithm~$\B$ for~$\P$,
the following algorithm (see Algorithm~\ref{algo:2stage-apx} for pseudocode) approximates $\M\restrict 2$ for $\M = (\P,\qualcap)$:

Given a 2-stage graph~$\G=(G_1,G_2)$,
we generate candidate $2$-stage solutions in a loop while storing the currently best overall solution throughout.
In the loop, with iterations indexed by $i=1,2,...$,
we consider a set~$Y$ that keeps track of \emph{intersection elements} which have not been in a solution for~$G_1$ in any previous iteration;
we initialize~$Y$ with~$X_\cap\coloneqq X_\P(G_1)\cap X_\P(G_2)$ and update $Y$ at the beginning of each iteration.
In iteration~$i$,
we use~$\B$ to find a solution $S_1^{(i)}\in f^*(G_1)$ that optimizes $\qualcap(S_1^{(i)}, Y)$;
we then use~$\B$ to find a solution~$S_2^{(i)}\in f^*(G_2)$ that optimizes~$\qualcap(S_1^{(i)},S_2^{(i)})$.
The loop stops as soon as~$Y$ is empty;
the output is the $2$-stage solution $(S_1,S_2)$ with maximum quality over all candidate solutions.

The approximation ratio depends on the input's intertwinement $\chi=\chi_\P(\G) = |X_\cap|$.

\begin{theorem}\label{thm:2-stage-apx}
Consider an MSP $\M = (\P,\qualcap)$ with preficient~$\P$.
Then, Algorithm~\ref{algo:2stage-apx} is a polynomial-time $1/\!\sqrt{2\chi}$-approximation algorithm for $\M\restrict 2$.
\end{theorem}

\begin{algorithm}[tb]
\caption{
	Approximation of $\M\restrict 2$ for~$\M = (\P, \qualcap)$ with preficient~$\P$
	\label{algo:2stage-apx}
}
\KwIn{Enriched $2$-stage graph $\G=(G_1,G_2)$, preficiency algorithm~\B for~$\P$}
\KwOut{$2$-stage solution $\S=(S_1,S_2)$}
$(S_1, S_2)\leftarrow (\varnothing, \varnothing)$\;
\For{$i = 1,2,...$}{
	$Y\leftarrow X_\cap\setminus\bigcup_{j\in\range{i-1}}S_1^{(j)}$\;
	\lIf{$Y = \varnothing$}{\Return{$(S_1, S_2)$}}
	$S_1^{(i)}\leftarrow \B(G_1,Y)$\;
	$S_{2}^{(i)}\leftarrow \B(G_2, S_1^{(i)})$\;
	\lIf{$\qualcap(S_1^{(i)}, S_2^{(i)}) \geq \qualcap(S_1, S_2)$}{%
		$(S_1, S_2)\leftarrow (S_1^{(i)}, S_2^{(i)})$}
}
\end{algorithm}

\begin{proof}\let\qed\relax
Let $\B$ be a preficiency algorithm for $\P$ and~$\G$ a 2-stage graph.
We can assume w.l.o.g.\ that $X_\cap$ is non-empty and thus $\opt\geq 1$.
Clearly, the first iteration establishes $\apx\geq 1$.

In each iteration~$i$ of the loop, at least one element of~$X_\cap$ that has not been in any previous first stage solution is contained in a solution for~$G_1$ (otherwise the loop terminates) and hence the loop terminates in polynomial time.
Let~$k$ denote the number of iterations.
For any~$i\in\range{k}$, let $(S_1^{(i)},S_2^{(i)})$ denote the 2-stage solution computed in the $i$th iteration.
Let~$(S^*_1, S^*_2)$ denote an optimal 2-stage solution and 	$S^*_\cap\coloneqq S^*_1\cap S^*_2\subseteq X_\cap$ its intersection
(note that $S^*_1\cap X_\cap\setminus S^*_\cap$ may be non-empty).
Let~\smash{$R_i \coloneqq (S_1^{(i)} \cap X_\cap) \setminus \bigcup_{j\in\range{i-1}}R_j$}
denote the set of intersection elements that are in $S_1^{(i)}$ but not in~$S_1^{(j)}$ for any previous iteration~$j < i$;
let~$r_i \coloneqq |R_i|$.
Note that in iteration~$i$,
the algorithm first searches for a solution $S_1^{(i)}\in f^*(G_1)$ that maximizes
\[
\qualcap(S_1^{(i)}, X_\cap \setminus\bigcup_{j\in\range{i-1}}S_1^{(j)})
= |S_1^{(i)}\cap \big( X_\cap\setminus\bigcup_{j\in\range{i-1}} R_j\big)| = r_i.\] 
We define
$R^*_i \coloneqq (S_1^{(i)} \cap S^*_\cap)\setminus\bigcup_{j\in\range{i-1}}R^*_j$ 
and $r^*_i\coloneqq |R^*_i|$
equivalently to~$R_i$, but w.r.t.~$S^*_\cap$ instead of~$X_\cap$ (cf.~\cref{fig:m-and-r}).
Thus, $R_i^*$ contains those elements of $S^*_\cap$ that are selected (into~$S_1^{(i)}$) for the first time over all iterations.
Observe that $R_i\cap S^*_\cap = R^*_i$.
\begin{figure}[tb]
	\centering
	\begin{tikzpicture}[text height=10pt,xscale=.9,yscale=-.7]
	\useasboundingbox[draw,blue] (0,0) rectangle (10,3);
	\tikzset{box/.style={rounded corners,thick}}

	\tikzset{firstBoxOuter/.style={box,red!20}}
	\tikzset{firstBox/.style={box,red!50}}
	\tikzset{firstBoxInner/.style={box,red!80}}

	\tikzset{secondBoxOuter/.style={box,\yelloworange!20}}
	\tikzset{secondBox/.style={box,\yelloworange!50}}
	\tikzset{secondBoxInner/.style={box,\yelloworange!80}}

	\tikzset{thirdBoxOuter/.style={box,mygreen!20}}
	\tikzset{thirdBox/.style={box,mygreen!50}}
	\tikzset{thirdBoxInner/.style={box,mygreen!80}}

	\newcommand{\eCapBox}{(0, 0) rectangle (10, 3.75)}
	\newcommand{\mCapBox}{(2, 0) rectangle (8, 3)}
	\newcommand{\firstMBox}{(-1.5,.25) rectangle (6,2.05)}
	\newcommand{\secondMBox}{(4,1) rectangle (7,4.55)}
	\newcommand{\thirdMBox}{(4.25,.5) rectangle (11.5,1.8)}

	\fill[box,\lightgray] \eCapBox;
	\fill[box,\darkgray] \mCapBox;

	\begin{scope}
		\fill[thirdBoxOuter] \thirdMBox;
		\clip \eCapBox;
		\fill[thirdBox] \thirdMBox;
		\clip \mCapBox;
		\fill[thirdBoxInner] \thirdMBox;
	\end{scope}

	\begin{scope}
		\clip \eCapBox;
		\fill[box,\lightgray] \secondMBox;
		\clip \mCapBox;
		\fill[box,\darkgray] \secondMBox;
	\end{scope}
	\begin{scope}
		\fill[secondBoxOuter] \secondMBox;
		\clip \eCapBox;
		\fill[secondBox] \secondMBox;
		\clip \mCapBox;
		\clip \secondMBox;
		\fill[secondBoxInner] \secondMBox;
	\end{scope}

	\begin{scope}
		\clip \eCapBox;
		\fill[box,\lightgray] \firstMBox;
		\clip \mCapBox;
		\fill[box,\darkgray] \firstMBox;
	\end{scope}
	\begin{scope}
		\fill[firstBoxOuter] \firstMBox;
		\clip \eCapBox;
		\fill[firstBox] \firstMBox;
		\clip \mCapBox;
		\clip \firstMBox;
		\fill[firstBoxInner] \firstMBox;
	\end{scope}

	\draw[box] \thirdMBox;
	\draw[box] \secondMBox;
	\draw[box] \firstMBox;
	\draw[box] \mCapBox;
	\draw[box] \eCapBox;

	\tikzset{label/.style={rectangle,rounded corners,fill=gray!1,inner sep=1pt,fill opacity=.8,text opacity=1}}
	\contourlength{0.06em}
	\newcommand{\xOffset}{0.33}
	\newcommand{\yOffset}{-0.38}
	\newcommand{\xOffsetS}{0.1}
	\newcommand{\yOffsetS}{0}
	\node[] at($(0,3.75)+(\xOffset,\yOffset)$) {\contour{white}{$X_\cap$}};
	\node[] at ($(2,3)+(\xOffset,\yOffset)$) {\contour{white}{$S^*_\cap$}};
	\node[] at ($(-1.5,2)+(\xOffsetS+\xOffset,\yOffsetS+\yOffset)$) {\contour{white}{$S^{(1)}_1$}};
	\node[] at ($(0,2)+(\xOffset,\yOffset)$) {\contour{white}{$R_1$}};
	\node[] at ($(2,2)+(\xOffset,\yOffset)$) {\contour{white}{$R^{*}_1$}};
	\node[] at ($(4,4.5)+(\xOffsetS+\xOffset,\yOffsetS+\yOffset)$) {\contour{white}{$S^{(2)}_1$}};
	\node[] at ($(4,3.75)+(\xOffset,\yOffset)$) {\contour{white}{$R_2$}};
	\node[] at ($(4,3)+(\xOffset,\yOffset)$) {\contour{white}{$R^{*}_2$}};
	\node[] at ($(11.5,1.75)+(-\xOffsetS-\xOffset,\yOffsetS+\yOffset)$) {\contour{white}{$S^{(3)}_1$}};
	\node[] at ($(10,1.75)+(-\xOffset,\yOffset)$) {\contour{white}{$R_3$}};
	\node[] at ($(8,1.75)+(-\xOffset,\yOffset)$) {\contour{white}{$R^{*}_3$}};
\end{tikzpicture}%
	\caption{%
		Visualization of the relationships between $X_\cap$, $S^*_\cap$, $S^{(i)}_1$, $R_i$ and $R^*_i$ for $i\in\range{3}$.\label{fig:m-and-r}
	}
\end{figure}
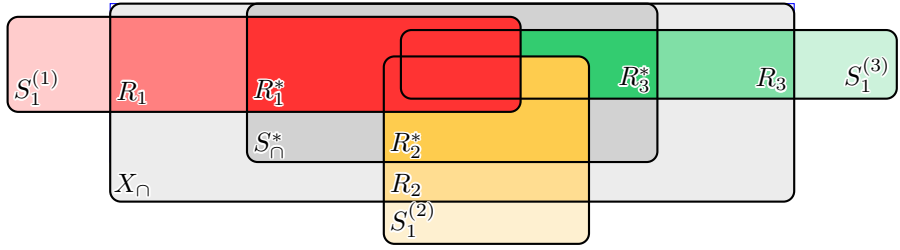

\newcommand{\iterpar}{x}%
Let $\iterpar \coloneqq \sqrt{2\chi}$.
For every $i\in\range{k}$ the algorithm chooses~$S_2^{(i)}$ such that $\qualcap(S_1^{(i)}, S_2^{(i)})$ is maximized.
Since we may choose $S_2^{(i)} = S^*_2$, it follows that $\apx \geq \max_{i\in\range{k}} r_i^*$.
Thus, if $\max_{i\in\range{k}} r_i^* \geq \opt/\iterpar$, we have a $1/\iterpar$-approximation.
In case $\opt\leq\iterpar$, any solution with profit at least~$1$ yields a $1/\iterpar$-approximation (which we trivially achieve as discussed above).
We show that we are always in one of the two cases.

\newcommand{\qceil}{\bar\iterpar}
Assume that $\opt > \iterpar$ and simultaneously $r_i^* < \opt/\iterpar$ for all~$i\in\range{k}$. In particular,~$i\geq 2$.
Since we distribute~$S^*_\cap$ over the disjoint sets $\{R^*_i\mid i\in\range{k}\}$, each containing less than~$\opt/\iterpar$ elements, we know that~$k > \iterpar$ (thus $k\geq\lceil \iterpar\rceil\eqqcolon\qceil$).
Recall that in iteration $i$, $Y=X_\cap\setminus\bigcup_{j\in\range{i-1}}S_1^{(j)}$.
Thus,
we have that $Y\cap \bigcup_{j\in\range{i-1}}R_j$ is empty and
the number of elements of $S^*_1$ that are counted towards $\qual(S^{(i)}_1,Y)$ is
$|(S^*_1 \cap X_\cap) \setminus \bigcup_{j\in\range{i-1}}R_j|
\geq |S^*_\cap \setminus \bigcup_{j\in\range{i-1}}R_j|
= |S^*_\cap \setminus \bigcup_{j\in\range{i-1}}R^*_j|$.
Hence, the latter term is a lower bound on~$r_i$ and we deduce:
\begin{linenomath}
		\begin{equation*}
			\begin{split}
				r_i&\geq \big|S^*_\cap\setminus \bigcup_{j\in\range{i-1}}R^*_j\big|
				= \opt -\sum_{j\in\range{i-1}}r^*_j\\
				&\stackrel{(\star)}\geq \opt - \sum_{j\in\range{i-1}} \frac{\opt}{\iterpar}
				= \opt\cdot\big(1 - \frac{i - 1}{\iterpar}\big)
				\geq \qceil\big(1 - \frac{i - 1}{\iterpar}\big).
			\end{split}
		\end{equation*}
\end{linenomath}
Thereby, strict inequality holds at $(\star)$ for $i\geq 2$.
This raises a contradiction:
\begin{linenomath}
		\begin{equation*}
			\begin{split}
				\chi
				&= |X_\cap|
				= \big|\bigcup_{i\in\range{k}} R_i\big|
				\geq
				\sum_{i\in\range{\qceil}} r_i
				\gneqq \sum_{i\in\range{\qceil}} \qceil \big(1 - \frac{i - 1}{\iterpar}\big)\\
				&= \qceil \cdot \big(\sum_{i\in\range{\qceil}}1 - \sum_{i\in\range{\qceil-1}}\frac{i}{\iterpar}\big)
				= \qceil \big(\qceil - \frac{(\qceil -1)\qceil}{2\iterpar}\big)\\
				&= \qceil^2 \big(1 - \frac{\qceil -1}{2\iterpar}\big)
				\geq \qceil^2 \big(1 - \frac{\iterpar}{2\iterpar}\big)
				= \frac{\qceil^2}{2}
				\geq \frac{\iterpar^2}{2}
				= \chi.
			\end{split}\hfill\qedsymbol
		\end{equation*}
\end{linenomath}
\end{proof}

We may mention that for the (preficient, see below) Multistage Perfect Matching problem there is (weak) evidence that a $\chi$-dependent ratio may be unavoidable~\cite{CTW22}.

\begin{observation}
The analysis in Theorem~\ref{thm:2-stage-apx} is \emph{tight} in the sense that Algorithm~\ref{algo:2stage-apx} cannot guarantee a better approximation ratio for arbitrary preficient SPs.
This is due to the fact that when the SP is the Perfect Matching Problem (see Appendix), we know from~\cite{CTW22} that there is an instance family for which the stripped down version of Algorithm~\ref{algo:2stage-apx} yields precisely this ratio.
This does not rule out that for some \enquote{simpler} SP, our algorithm may achieve a better approximation ratio.
\end{observation}


\section{Applications}\label{sec:applications}

We show preficiency for a variety of SPs.
Proving preficiency typically always follows the same pattern for these problems:
we modify (or assign) weights for those graph elements that are to be preferred
in a way that does not interfere with the feasibility of a solution;
we then apply a polynomial algorithm (as a black box) to solve the problem w.r.t.\ the modified weights.
The above \cref{thm:2-stage-apx,thm:tim-to-mim} then allow us to deduce approximation algorithms for the corresponding $\qualcap$-MSP.

In general, instead of manipulating the weights, one could (try to) carefully manipulate the arithmetic computations in the black box algorithm.
This then would typically also allow to consider non-negative weights (instead of strictly positive ones).
However, we refrain from doing so herein for clarity of exposition and general applicability to any black box algorithm where such arithmetic modifications may not be straightforward.


Let~$G=(V,E)$ be a graph with weights~$w\colon X\to\mathbb N_{>0}$ on its elements.
Let~$w(Z)\coloneqq \sum_{e\in Z}w(e)$ denote the weight of a subset~$Z\subseteq X$.
Given an element subset~$Y\subseteq X$ and some~$\varepsilon\in\mathbb Q$,
we define the \emph{modified weight function}~$\wmod\colon X\to\mathbb Q$
that is identical to~$w$ on~$X\setminus Y$
and~$\wmod(e)\coloneqq w(e)-\varepsilon$ for~$e\in Y$.

Consider some SP~$P$.
The modified weight function~$\wmod$ is \emph{well-behaved w.r.t.~$P$} if the following properties hold for any two edge sets~$Z,Z'\subseteq X$:
\begin{enumerate}[(i)]
	\item $\wmod(e)>0$ for all~$e\in X$;
	\item if~$w(Z')< w(Z)$, then~$\wmod(Z') < \wmod(Z)$;
	\item if~$w(Z') = w(Z)$ and~$|Z'\cap Y| > |Z\cap Y|$,
	then: if~$P$ is a minimization problem then $\wmod(Z') < \wmod(Z)$; otherwise ($P$ is a maximization problem) $\wmod(Z')>\wmod(Z)$.
\end{enumerate}
Naturally, $\varepsilon>0$ for minimization problems and $\varepsilon<0$ for maximization problems.

\newcommand{\thmtext}[1]{
There is a $1/\!\sqrt{8\chi}$-approximation algorithm for #1 and a $1/\!\sqrt{2\chi}$-approximation for $#1\restrict2$.
}


\subsection{Multistage Shortest \st Path}

As a first example, consider the classic problem of finding a shortest \st path, which can be easily formulated as an SP.
The corresponding MSP is introduced and shown to be NP-hard already for $2$-stage DAGs in~\cite{FNSZ20} (although they consider a slightly different definition, the NP-hardness proof directly translates to our formulation).

In a graph $G=(V,E)$ with edge weights $w\colon E\to \mathbb N_{>0}$ and two terminal vertices $s,t\in V$,
an edge set~$F\subseteq E$ is an \emph{\st path} in $G$
if it is of the form $\{v_1v_2, v_2v_3, ..., v_{k-1}v_k\}$ where $k\geq 2$, $v_1=s$ and $v_k=t$.
An \st path is a \emph{shortest \st path}
if for each \st path $F'$ we have $w(F)\leq w(F')$.

\begin{definition}[\MSPath]
Given a multistage graph~$(G_i)_{i\in\range\tau}$ with edge weights $w_i\colon E_i\to \mathbb N_{>0}$ and terminal vertices $s_i,t_i\in V_i$ for each~$i\in\range\tau$,
find a $\qualcap$-optimal multistage solution $(F_i)_{i\in\range\tau}$ such that for each $i\in\range\tau$,
$F_i\subseteq E_i$ is a shortest path from $s_i$ to $t_i$ in~$G_i$.
\end{definition}

\begin{theorem}\label{thm:MSPath-algo}
\thmtext{\MSPath}
\end{theorem}
\begin{proof}
We only need to show preficiency for~\MSPath and apply Theorems~\ref{thm:tim-to-mim} and~\ref{thm:2-stage-apx}.
Let $G=(V,E)$ be a graph with edge weights $w\colon E\to \mathbb N_{>0}$ and $Y\subseteq E$ the set of edges to be preferred.
Set $\varepsilon\coloneqq 1/(|E|+1)$.
We use an arbitrary polynomial-time algorithm for computing a shortest \st path in $G$ with the modified weight function~$\wmod$, e.g.~\cite{Dij59}, and denote its output by~$F$.
Since $\wmod$ is well-behaved, $F$ is a shortest \st path in~$G$ such that~$|F\cap Y|$ is maximum.
\end{proof}


\subsection{Further examples}
The brevity of the above proof under the new framework is no exception, as we were able to apply nearly identical proofs for a range of other $\qualcap$-MSP formulations of classical combinatorial problems, which we give as further examples of our framework's utility.
To justify the need for approximation algorithms, proofs that the presented multistage problems are NP-hard (even when restricted to two stages) can be found in the appendix. 
There, one can also find the precise problem definitions, variations, and preficiency proofs.
We stress that there are no known constant-ratio approximations for any of the problems.

\begin{theorem}\label{thm:meta-thm}
Consider the MSP~$\M=(\P,\qualcap)$, where $\P$ is Minimum Weight Perfect Matching, Minimum \st Cut, Weakly Bipartite Maximum Cut, Minimum Weight Bipartite Vertex Cover, or Maximum Weight Bipartite Independent Set. There is a $1/\!\sqrt{8\chi}$-approximation algorithm for~$\M$ and a $1/\!\sqrt{2\chi}$-approximation for $\M\restrict2$.
\end{theorem}



\section{Conclusion}

In this paper, we considered multistage generalizations of Subgraph Problems that require an optimal solution in each individual stage while the transition quality is to be optimized.
We provided two framework approximation algorithms for such MSPs: Algorithm~\ref{algo:tm-to-mim} allows to generalize any $2$-stage algorithm for any MSP to an unrestricted number of stages; Algorithm~\ref{algo:2stage-apx} is a $2$-stage approximation algorithm for any preficient $\qualcap$-MSP.
We then showcased the ease-of-use of our results by applying them to several natural MSP variants of well-known classical graph problems.

It remains open, whether these algorithms are best possible for any of the considered MSPs.
In fact, there \emph{cannot} be a general result establishing tightness for the whole class of MSPs, as some problems are actually polynomial-time solvable (see, e.g., the vertex variants of \MMinCE above).
For ease of exposition, we have only considered multistage generalizations of \emph{subgraph} problems in this paper.
However, our techniques are also applicable to more general multistage \emph{subset} problems, i.e., without the need of an underlying graph.
This can be easily understood as all our proofs solely work on a set system on ground set $X$.
Alas, we know of no natural multistage non-subgraph optimization problem that simultaneously is (a)~NP-hard, (b)~preficient, and (c)~not trivially reformulated as an MSP.

For a further investigation on the interdependency between the two concurring multistage optimization objectives, it might be of interest to consider another natural edge case: require optimal transition quality for each transition while allowing suboptimal, but feasible per-stage solutions.


\bibliography{main2.bib}

\clearpage
\appendix


\section{Multistage Minimum (Weight) Perfect Matching}\label{subsection:MMinPM}
In a graph $G=(V,E)$ with edge weights $w\colon E\to \mathbb N_{>0}$,
an edge set $F\subseteq E$ is a \emph{perfect matching}
if each $v\in V$ is incident with exactly one edge in $F$.
A perfect matching $F$ has \emph{minimum weight}
if for each perfect matching $F'$ we have $w(F')\geq w(F)$.

\begin{definition}[\MMinWPM]
Given a multistage graph~$(G_i)_{i\in\range\tau}$ with edge weights $w_i\colon E_i\to \mathbb N_{>0}$ for each $i\in\range\tau$,
find a $\qualcap$-optimal multistage solution $(F_i)_{i\in\range\tau}$ such that for each~$i\in\range\tau$,
$F_i\subseteq E_i$ is a minimum weight perfect matching w.r.t.~$w_i$ in~$G_i$.
\end{definition}

$\MMinWPM\restrict{2}$ is NP-hard even if the union of all stages is bipartite~\cite{CTW22}.

\begin{theorem}\label{thm:MMinPM-algo}
\thmtext{\MMinWPM}
\end{theorem}

\begin{proof}
Follow the proof of Theorem~\ref{thm:MSPath-algo} using, e.g., \cite{LP86} for the efficient minimum weight perfect matching computation.
\end{proof}

This result does not contradict the linear lower bound on the approximation ratio discussed in~\cite{BELP18}, since they (a) minimize an objective function combining transition costs and per-stage solution quality and (b) do not consider intertwinement dependency. Note that there are graphs with linear intertwinement~$\chi=\Theta(|E|)$.


\section{Multistage Minimum \st Cut}
In a graph $G=(V,E)$ with edge weights $w\colon E\to \mathbb N_{>0}$, two vertices $s,t\in V$,
an edge set $F\subseteq E$ is an \emph{\st cut}
if there is no \st path in~$(V, E\setminus F)$.
An \st cut $F$ is \emph{minimum} if for each \st cut $F'$ we have $w(F') \geq w(F)$.

\begin{definition}[\MMinCE]
Given a multistage graph $(G_i)_{i\in\range\tau}$ with edge weights $w_i\colon E_i\to \mathbb N_{>0}$
and terminal vertices $s_i,t_i\in V_i$ for each $i\in\range\tau$,
find a $\qualcap$-optimal multistage solution~$(F_i)_{i\in\range\tau}$ such that for each~$i\in\range\tau$,
$F_i\subseteq E_i$ is a minimum \siti cut for $s_i,t_i$ in $G_i$.
\end{definition}

\begin{theorem}\label{thm:MMinCutSet-hardness}
$\MMinCE\restrict2$ is NP-hard, even if $s_1=s_2$, $t_1=t_2$, and the edges have uniform weights.
\end{theorem}

\begin{proof}
We will perform a reduction from the NP-hard problem \textsf{MaxCut}~\cite{GJ79} to the decision variant of \MMinCE: Given $G$, $(s_i,t_i)_{i\in\range\tau}$ and a number $\kappa\in\mathbb N$, is there a $\qualcap$-optimal multistage solution for \MMinCE with profit at least $\kappa$?
In \textsf{MaxCut}, one is given an undirected graph $G=(V,E)$ and a number $k\in\mathbb N$; the question is whether there is a vertex set $U\subseteq V$, such that $|\delta(U)|\geq k$.
In the first stage, we will construct a bundle of \st paths for each vertex of the original graph and in the second stage we will create two \st paths for each edge~(cf. \cref{fig:MMinCutSet}).
A minimum \st cut in the first stage will correspond to a vertex selection and a minimum \st cut in the second stage will allow us to count the number of edges that are incident to exactly one selected vertex.

Given an instance $\mathcal I \coloneqq \big(G=(V,E),k\big)$ of \textsf{MaxCut}, we will construct an equivalent instance $\mathcal J \coloneqq \big(\G, (s_i, t_i)_{i\in\range\tau}, \kappa\big)$ for \MMinCE.
Set $\kappa \coloneqq |E| + k$.
Start with a 2-stage graph~$\G\coloneqq (G_1,G_2)$ and vertices $s,t\in V_1\cap V_2$ that are used as terminals $s_i$ and $t_i$ in both stages.

In the first stage, create a vertex $x_v$ for each $v\in V$.
For each $v\in V$ and for each edge~$e\in\delta(v) \coloneqq \{e\in E\mid v\in e\}$, create a path of length~$3$ from $s$ to $x_v$ whose middle edge is labeled~$\alpha_v^e$ and a path of length~$3$ from~$x_v$ to~$t$ whose middle edge is labeled $\beta_v^e$.
Let $A_v\coloneqq \{\alpha_v^e\mid e\in\delta(v)\}$ and~$B_v\coloneqq \{\beta_v^e\mid e\in\delta(v)\}$.
Let $a^e_v$ ($b^e_v$) denote the endpoint of $\alpha^e_v$ ($\beta^e_v$) closer to~$x_v$ and $\bar a^e_v$ ($\bar b^e_v$) the other one.

The second stage reuses all $a$-, $\bar a$-, $b$- and $\bar b$-vertices and all $\alpha$- and $\beta$-edges.
For each edge $e = vw\in E$, add two vertices~$c^e_v, c^e_w$.
By adding edges $\{s c^e_w, c^e_w \bar a^e_w, a^e_w \bar a^e_v, a^e_v t\}$ and~$\{s b^e_w, \bar b^e_w b^e_v, \bar b^e_v c^e_v, c^e_v t\}$,
we construct two paths, $A^e$ and $B^e$, of length~$6$ from $s$ to $t$:
$A^e$ via the~$\alpha$-edges and $B^e$ via the $\beta$-edges.
The $c$-vertices are there to avoid unwanted edges in~$E_\cap$.

\begin{figure}
	\begin{center}
		\includegraphics[scale=.99]{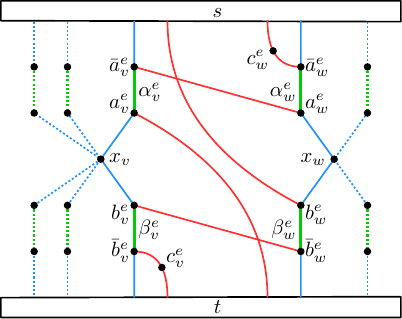}
		\caption{Thm.~\ref{thm:MMinCutSet-hardness}: Two vertex gadgets and one edge gadget for $e=vw$.
			Vertices $s$ and $t$ are enlarged.
			Edges in~$E_\cap$ are bold and green, edges in~$E_1\setminus E_\cap$ are blue, edges in~$E_2\setminus E_\cap$ are red.}\label{fig:MMinCutSet}
	\end{center}
\end{figure}

Since for each $v\in V$, a minimum cut $C_1$ in $G_1$ needs to cut all \st paths through~$x_v$, $C_1$ contains exactly $|\delta(v)|$ edges from these paths;
these are either all part of the $3$-paths from $s$ to $x_v$ or all part of the $3$-paths from $x_v$ to $t$.
We show that $\mathcal J$ is a yes-instance if and only if~$\mathcal I$ is a yes-instance:
\smallskip

\noindent "$\Leftarrow$" \tabto{10mm}%
Suppose there is some $S\subseteq V$, such that $|\delta(S)|\geq k$.
For each $v\in S$, add $A_v$ to $C_1$ and for each $v\in V\setminus S$, add $B_v$ to $C_1$.
Thus, $C_1$ is a minimum \st cut in $G_1$.
For~$e = vw\in \delta(S)$ assume w.l.o.g.\ $v\in S$ and add $\alpha_v^e$ and $\beta_w^e$ to $C_2$.
For~$e = vw\in E\setminus \delta(S)$, add arbitrarily either~$\{\alpha_v^e,\beta_w^e\}$ or $\{\alpha^e_w,\beta^e_v\}$ to $C_2$.
Thus $C_2$ is a minimum \st cut in $G_2$.

Consider some edge $e=vw\in E$.
If~$e\in\delta(S)$, assume w.l.o.g.\ $v\in S$ and $w\in V\setminus S$.
Then, $C_1\cap C_2$ contains two edges of $A^e\cup B^e$, namely $\alpha_v^e$ and $\beta_w^e$.
If $e\not\in \delta(S)$, $C_1\cap C_2$ contains exactly one edge of $A^e\cup B^e$.
Thus, $|C_1\cap C_2| = |\biguplus_{e\in E} C_1\cap C_2\cap (A^e\cup B^e)| = |E| + |\delta(S)|$.
\smallskip

\noindent "$\Rightarrow$" \tabto{10mm}%
Let $\C \coloneqq (C_1,C_2)$ be a multistage minimum \st cut with $|C_1\cap C_2| \geq |E| + k$.
Observe that, by construction of~$G_2$, $|C_1\cap C_2| = \sum_{e\in E}m_e$ with $m_e \coloneqq|C_1\cap C_2\cap (A^e\cup B^e)|\leq 2$ for each $e\in E$.
Thus, by pigeonhole principle, there are at least $k$ edges with $m_e= 2$.
Recall that w.l.o.g.\ and by optimality of the individual stages, we can assume for each $v\in V$ that~$C_1$ contains either all of $A_v$ or all of $B_v$ but no elements of the other.
This yields a selection~$U\subseteq V$: Select $v\in V$ into~$U$ if and only if $A_v\subseteq C_1$.
We observe that $m_e = 2$ then induces~$e\in \delta(U)$ and we obtain $|\delta(S)|\geq k$.
\end{proof}

\begin{theorem}
\thmtext{\MMinCE}
\end{theorem}
\begin{proof}
Follow the proof of Theorem~\ref{thm:MSPath-algo} using, e.g., \cite{FF56} for the efficient minimum \st cut computation.
\end{proof}


\subparagraph{Vertex variants.}
The problem of finding a minimum \st cut for each stage can also be optimized to maintain the same set of vertices that are connected to~$s$.
For a concise problem definition, we need new terminology:
in a graph $G=(V,E)$ with edge weights $w\colon E\to \mathbb N_{>0}$ and vertices $s,t\in V$,
a vertex set $S\subseteq V$ with $s\in S$, $t\not\in S$ is an \emph{\st separating partition};
$S$~is \emph{optimal}, if the induced \st cut~$\delta(S)\coloneqq\{e\in E\mid |e\cap S|=1\}$ has minimum weight.

Given a multistage graph $(G_i)_{i\in\range\tau}$ with edge weights $w_i\colon E_i\to \mathbb N_{>0}$ for each $i\in\range\tau$ and terminal vertices $s_i,t_i\in V_i$ for each $i\in\range\tau$,
we may ask for a $\qualcap$-optimal multistage solution~$(S_i)_{i\in\range\tau}$ such that for each~$i\in\range\tau$,
$S_i\subseteq V_i$ is an optimal \siti separating partition for~$s_i,t_i$ in~$G_i$.
The objective is to maximize the global quality w.r.t.\ intersection profit.
A natural variation is to consider $q'\coloneqq |S_i\cap S_{i+1}| + |(V_i\setminus S_i)\cap (V_{i+1}\setminus S_{i+1})|$ instead of~$\qualcap$.

\begin{observation}
Both vertex variants of \MMinCE (using~$\qualcap$ or $q'$, resp.) are polynomial-time solvable.
\end{observation}
\begin{proof}
Consider~$\qualcap$.
For each stage~$i$, choose the cardinality-wise largest optimal \siti separating partition~$S_i$.
Observe that~$S_i$ with $s_i\in S_i$ is unique, since every other optimal \siti separating partition~$S'_i$ is a strict subset of~$S_i$.
Clearly, $S_i$ can be found in polynomial time.
Since we optimize the intersection of the stage-wise partitions, we obtain a global optimum by having each $S_i$ maximal.

Consider~$q'$.
The problem can be easily reduced to a single-stage minimum \st cut problem as shown in~\cite{BEK21}:
Add disjoint copies of each stage to an empty graph, and two new vertices~$s^*$ and~$t^*$.
Add an edge with infinite weight from~$s^*$ to each~$s_i$ and one from each~$t_i$ to~$t^*$.
For each occurrence of a vertex in two adjacent stages, add an edge with small positive weight~$\varepsilon$ between the two vertex copies.
A minimum $s^*$-$t^*$-cut in this graph directly induces a minimum \siti cut in each stage such that the number of vertices that are in~$(S_i\cup S_{i+1})\setminus(S_i\cap S_{i+1})$ is minimized and thus $q'$ is maximized.
\end{proof}


\section{Multistage Weakly Bipartite Maximum Cut}
In a graph $G=(V,E)$ with edge weights $w\colon E\to \mathbb N_{>0}$,
a vertex set $U\subseteq V$ induces a \emph{maximum cut} $\delta(U)\coloneqq \{e\in E\mid |E\cap U|=1\}$
if for each vertex set $U'\subseteq V$ we have $w\big(\delta(U')\big)\leq w\big(\delta(U)\big)$.
The (unfortunately named) class of \emph{weakly bipartite graphs} is defined in~\cite{GP81} and contains in particular also planar and bipartite graphs.
While the precise definition is not particularly interesting to us here,
we make use of the fact that a maximum cut can be computed in polynomial time on weakly bipartite graphs~\cite{GP81}.

\begin{definition}[\MMaxWBC]
Given a multistage graph $(G_i)_{i\in\range\tau}$ with edge weights $w_i\colon E_i\to \mathbb N_{>0}$ for each $i\in\range\tau$ where each stage is weakly bipartite,
find a $\qualcap$-optimal multistage solution $(F_i)_{i\in\range\tau}$ such that for each $i\in\range\tau$,
$F_i\subseteq E_i$ is a maximum cut in $G_i$.
\end{definition}

\begin{theorem}\label{thm:MMaxWBC-hardness}
$\MMaxWBC\restrict2$ is NP-hard already on multistage graphs where each stage is planar.
\end{theorem}

\begin{theorem}
\thmtext{\MMaxWBC}
\end{theorem}
\begin{proof}
Follow the proof of Theorem~\ref{thm:MSPath-algo} but using $\varepsilon\coloneqq -1/(|E|+1)$ and, e.g., \cite{GP81} for the efficient maximum cut computation.
\end{proof}


\section{Multistage Minimum Weight Bipartite Vertex Cover}
In a bipartite graph $G=(V,E)$ with vertex weights~$w\colon V\to\mathbb N_{>0}$,
a vertex set $U\subseteq V$ is a \emph{vertex cover}
if each $e\in E$ is incident with at least one vertex in~$U$.
A vertex cover $U$ has \emph{minimum weight}
if for each vertex cover $U'$ we have $w(U)\leq w(U')$.
\MMinBVC aims to maximize the number of common vertices:

\begin{definition}[\MMinBVC]
Given a multistage graph $(G_i)_{i\in\range\tau}$ with vertex weights $w_i\colon V_i\to \mathbb N_{>0}$ for each $i\in\range\tau$ where each stage is bipartite,
find a $\qualcap$-optimal multistage solution~$(U_i)_{i\in\range\tau}$ such that for each~$i\in\range\tau$,
$U_i\subseteq V_i$ is a minimum weight vertex cover in~$G_i$.
\end{definition}

\begin{theorem}\label{thm:MMinBVC-hardness}
$\MMinBVC\restrict2$ is NP-hard already with uniform weights on multistage graphs where each stage only consists of disjoint cycles.
\end{theorem}

\begin{proof}
We will perform a reduction from the unweighted \textsf{MaxCut} problem on graphs with maximum degree~$3$~\cite{Yan78} to the (unweighted) decision variant of \MMinBVC:
Given a two-stage graph~$\G$ where each stage is bipartite and a number~$\kappa\in\mathbb N$, is there a $(\qualcap,\max)$-optimal multistage solution $(U_1,U_2)$ for \MMinBVC with intersection profit $\qualcap(U_1,U_2)\geq\kappa$?
In \textsf{MaxCut}, one is given an undirected graph~$G=(V,E)$ and a natural number $k$; the question is to decide whether there is an $S\subseteq V$ such that $|\delta(S)|\geq k$.

Given an instance $\mathcal I \coloneqq (G=(V,E),k)$ of \textsf{MaxCut}, we construct an equivalent instance~$\mathcal J \coloneqq (\G, \kappa)$ of \MMinBVC.
Set $\kappa \coloneqq |E| + k$.
We start with an empty 2-stage graph~$\G\coloneqq (G_1,G_2)$.
In~$G_1$, we will have disjoint gadgets for each vertex, allowing three incident edges each.
In~$G_2$, we will have disjoint gadgets for each edge, intersecting with the two corresponding vertex gadgets in~$G_1$ (cf.~\cref{fig:MMinBVC-hardness}).
A minimum vertex cover in~$G_1$ will correspond to a vertex selection in~$G$
and a minimum vertex cover in~$G_2$ will allow us to count the edges of~$G$ that are incident to exactly one selected vertex.

\begin{figure}
	\begin{center}
		\includegraphics[scale=1]{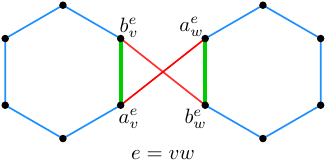}
		\caption{
			Thm.~\ref{thm:MMinBVC-hardness}:
			Bold green edges are in~$E_\cap$, blue edges in $E_1\setminus E_\cap$, red edges in~$E_2\setminus E_\cap$.
		}\label{fig:MMinBVC-hardness}
	\end{center}
\end{figure}

In $G_1$, for each $v\in V$ we create a $6$-cycle~$C_v$ and label its vertices counter-clockwise with $a_v^e$, $b_v^e$, $a_v^{e'}$, $b_v^{e'}$, $a_v^{e''}$, $b_v^{e''}$,
where $\{e,e',e''\} = \delta(v)$ denote the edges incident with $v$
(if $|\delta(v)|<3$, the vertices are labeled with fictional indices).
In $G_2$, we use the same vertex set and for each $e=vw\in E$ we create the $4$-cycle~$C_e$ by introducing the edges $a_v^e a_w^e$, $a_w^e b_w^e$, $b_w^e b_v^e$, $b_v^e a_v^e$.

For a minimum vertex cover~$U_1$ in~$G_1$, for each~$C_v$ either all $a$-vertices or all $b$-vertices are in~$U_1$.
For a minimum vertex cover~$U_2$ in~$G_2$, for each~$C_e$ with $e=vw$, either $a_v^e, b_w^e\in U_1$ or $a_w^e, b_v^e\in U_1$.
We show that~$\mathcal J$ is a yes-instance if and only if~$\mathcal I$ is a yes-instance:
\smallskip

\noindent \enquote{$\Leftarrow$} \tabto{10mm}
Suppose there is an~$S\subseteq V$ such that~$|\delta(S)|\geq k$.
We construct a vertex cover~$U_1$ for~$G_1$ as follows:
for each~$v\in S$, pick all $a$-vertices of~$C_v$ into~$U_1$;
for each~$v\in V\setminus S$, pick all~$b$-vertices of~$C_v$ into~$U_1$.
Similarly, we construct a vertex cover~$U_2$ for~$G_2$:
for each~$e = vw$ with $v\in S$ and $w\not\in S$ (i.e., $e\in\delta(S)$), pick~$a_v^e$ and~$b_w^e$ into~$U_2$;
for each~$e\not\in\delta(S)$, pick two arbitrary opposite vertices in~$C_e$.

Consider some edge $e=vw\in E$.
If $e\in\delta(S)$, $U_1\cap C_e = U_2\cap C_e$ and two vertices are chosen identically.
If $e\not\in\delta(S)$, $U_1\cap C_v$ and $U_1\cap C_w$ contain either both an $a$-vertex or both a $b$-vertex; since~$U_2$ must contain one vertex of each type, exactly one vertex is chosen identically.
Summing over all edge cycles yields~$\qualcap(U_1,U_2) = |E| + |\delta(S)|\geq \kappa$.
\smallskip

\noindent \enquote{$\Rightarrow$} \tabto{10mm}
Let $(U_1,U_2)$ be a $(\qualcap,\max)$-optimal multistage solution with $\qualcap(U_1,U_2) \geq |E| + k$.
We construct a vertex selection~$S\subseteq V$ in the original graph according to the following rule: pick $v$ into $S$ if and only if the $a$-vertices of~$C_v$ are in~$U_1$.

Consider some edge $e=vw\in E$.
If~$v\in S$ and $w\not\in S$ (i.e., $e\in\delta(S)$),
$U_1\cap C_e = \{a_v^e, b_w^e\}$.
Since $U_2$ maximizes the intersection with $U_1$ and can be chosen independently on each~$C_e$,
$U_2$ must be identical to $U_1$ on $C_e$; thus, $|U_1\cap U_2\cap C_e| = 2$.
If~$e\not\in\delta(S)$,
every minimum vertex cover~$U_2$ can contain at most one vertex of $U_1\cap C_e$.
Summing up $U_1\cap U_2$ over all edge gadgets,
we have $\kappa\leq |U_1\cap U_2| = 2\cdot |\delta(S)| + |E\setminus\delta(S)|$ and thus $|\delta(S)|\geq k$.
\end{proof}

\begin{theorem}\label{thm:MMinBVC-algo}
\thmtext{\MMinBVC}
\end{theorem}
\begin{proof}
We only need to show preficiency for~\MMinBVC and apply Theorems~\ref{thm:tim-to-mim} and~\ref{thm:2-stage-apx}.
Let $G=\big(V=(A,B), E\big)$ be a bipartite graph with vertex weights~$w\colon V\to\mathbb N_{>0}$ and $Y\subseteq V$ the set of vertices to be preferred.
Let $\varepsilon\coloneqq 1/(|V|+1)$.
Construct the modified graph $G'$ from~$G$ by adding two new vertices $s,t$ and edge sets $\{sv\mid v\in A\}$ and $\{vt\mid v\in B\}$.
We equip $G'$ with edge weights $w'(uv)\coloneqq w(v)-\varepsilon\cdot \1{v\in Y}$ for $u\in\{s,t\}$, and $w'(uv)\coloneqq\infty$ otherwise.

Note that $w'$ is well-behaved w.r.t.\ $w$.
It is well-known that a minimum weight \st cut~$C\subseteq E$ in~$G'$ (computable in polynomial time~\cite{FF56}) induces a minimum weight vertex cover~$U$ in~$G$ by picking all vertices $v\in V$ that are incident with an edge in~$C$.
Further, $U$ maximizes $|U\cap Y|$:
Suppose there is a minimum weight vertex cover $U'$ in $G$ with $|U'\cap Y| > |U\cap Y|$.
Let $C'$ again denote the \st cut associated with~$U'$.
By construction and since $w(U') = w(U)$,
we have $w'(C') < w'(C)$, again contradicting minimality of $C$ w.r.t.~$w'$.
\end{proof}


\section{Multistage Maximum Weight Bipartite Independent Set}
In a graph $G=(V,E)$ with vertex weights~$w\colon V\to\mathbb N_{>0}$,
a vertex set $U\subseteq V$ is an \emph{independent set}
if for $u,v\in U$ with $u\neq v$ we have $uv\not\in E$.
An independent set $U$ has \emph{maximum weight}
if for each independent set~$U'$ we have $w(U)\geq w(U')$.

\begin{definition}[\MMaxBIS]
Given a multistage graph $(G_i)_{i\in\range\tau}$ with vertex weights $w_i\colon V_i\to \mathbb N_{>0}$ for each $i\in\range\tau$ where each stage is bipartite,
find a $\qualcap$-optimal multistage solution $(U_i)_{i\in\range\tau}$ such that for each~$i\in\range\tau$,
$U_i\subseteq V_i$ is a maximum weight independent set in~$G_i$.
\end{definition}

It is well-known that the complement of a minimum weight vertex cover is a maximum weight independent set.
However, the complement of an optimal multistage vertex cover in general does not yield an optimal multistage independent set.
Nonetheless, the former property is still the key to obtain the following results.

\begin{theorem}\label{thm:MMaxBIS-hardness}
$\MMaxBIS\restrict2$ is NP-hard already with uniform weights on multistage graphs where each stage only consists of disjoint cycles.
\end{theorem}
\begin{proof}
Follow the proof for \cref{thm:MMinBVC-hardness}.
Equivalently to before, every cycle in every stage allows exactly two (sub)solutions.
In particular, these are the very same subsolutions as before, since here maximum independent sets are identical to minimum vertex covers.
To solve \textsc{MaxCut}, one has to pick subsolutions with maximum intersection.
\end{proof}

\begin{theorem}
\thmtext{\MMaxBIS}
\end{theorem}
\begin{proof}
Follow the proof of \cref{thm:MMinBVC-algo}, but using~$\varepsilon\coloneqq -1/(|V|+1)$.
Now, selecting the complement of the vertex cover yields an independent set with the maximum number of vertices from~$Y$.
\end{proof}
\end{document}